\documentclass{elsarticle}
\usepackage[bookmarks]{hyperref}
\usepackage{
	amsfonts,
	amsmath,
	amssymb,
	amsthm,
	array,
	caption,
	chngpage,
	color,
	float,
	graphicx,
	latexsym,
	mathrsfs,
	multicol,
	multirow,
	spreadtab,
	rotating,
	subcaption,
	tikz,
	tikz-cd,
	url,
	verbatim,
	wrapfig,
	wasysym
}

\journal{Discrete Mathematics}

\usetikzlibrary{arrows,automata, positioning}
\usetikzlibrary{decorations.pathmorphing}
\tikzset{snake it/.style={decorate, decoration=snake}}

\newcommand{\abs}[1]{\lvert#1\rvert}
\newcommand{\floor}[1]{\left\lfloor#1\right\rfloor}
\newcommand{\ceil}[1]{\left\lceil#1\right\rceil}

\newtheorem{thm}{Theorem}
\newtheorem{cor}[thm]{Corollary}
\newtheorem{lem}[thm]{Lemma}

\newtheorem{con}[thm]{Conjecture}
\newtheorem{df}[thm]{Definition}

\begin{document}
	\begin{frontmatter}
	\title{
		Automatic complexity of shift register sequences
	}
	\author{
		Bj{\o}rn Kjos-Hanssen
	}
	\address{University of Hawai\textquoteleft i at M\=anoa}
	\ead[url]{https://math.hawaii.edu/wordpress/bjoern/}
	\begin{abstract}
		Let $x$ be an $m$-sequence, a maximal length sequence produced by a linear feedback shift register. We show that $x$ has maximal subword complexity function in the sense of Allouche and Shallit. We show that this implies that the nondeterministic automatic complexity $A_N(x)$ is close to maximal: $n/2-A_N(x)=O(\log^2n)$, where $n$ is the length of $x$.
		In contrast, Hyde has shown $A_N(y)\le n/2+1$ for all sequences $y$ of length $n$.
	\end{abstract}
	
\end{frontmatter}


	\section{Introduction}
		Linear feedback shift registers, investigated and popularized by Golomb \cite{MR0242575}, may be
		``the most-used mathematical algorithm idea in history'',
		used at least $10^{27}$ times in cell phones and other devices \cite{WinNT}.
		They are particularly known as a simple way of producing pseudorandom output sequences called $m$-sequences.
		However, thanks to the Berlekamp--Massey algorithm \cite{BerlekampMassey},
		one can easily find the shortest LFSR that can produce a given sequence $x$.
		The length of this LFSR, the \emph{linear complexity} of $x$, 
		should then be large for a truly pseudorandom sequence, but is small for $m$-sequences.
		In this article we show that using a different complexity measure, \emph{automatic complexity},
		the pseudorandomness of $m$-sequences can be measured and, indeed, verified.

		Roughly speaking, finite automata are not able to detect significant patterns in shift register sequences.
		Moreover, shift register sequences seem to give an answer to the question
		\begin{quote}
			``What kind of sequences have high automatic complexity?''
		\end{quote}
		See in particular some results of computer experimentation in Section \ref{computer}.

	\section{Definitions}
		While our computer results in Section \ref{computer} concern the linear case specifically,
		our theoretical results in Section \ref{theory} concern the following natural abstraction of
		the usual notion of feedback shift register \cite{GammelG}.
		\begin{df}
			Let $q$ be a positive integer and let $[q]=\{0,\dots,q-1\}$.
			A $q$-ary $k$-stage 
			combinatorial shift register (CSR) is a mapping
			\[
				\Lambda : [q]^k \to [q]^k
			\]
			such that there exists $F:[q]^k\to [q]$ such that for all $x_i$,
			\[
				\Lambda(x_0,\dots,x_{k-1}) = (x_1,x_2,\dots,x_{k-1},F(x_0,x_1,\dots,x_{k-1})).
			\]
			The function $F$ is called the \emph{feedback function} of $\Lambda$.
		\end{df}
		\begin{df}
			An infinite sequence $x=x_0x_1\dots$ is \emph{eventually periodic} if
			there exist integers $M$ and $N>0$ such that for all $n>M$,
			$x_n=x_{n-N}$.
			The least $N$ for which there exists such an $M$ is the \emph{period} of $x$.
		\end{df}
		\begin{df}
			For any $k$-stage CSR $\Lambda$ and any word $x$ of length $\ge k$, the
			\emph{period of $\Lambda$ upon processing $x$} is the period of the sequence $\Lambda^t(x_0,\dots,x_{k-1})$, $0\le t<\infty$.
		\end{df}
		Lemma \ref{fire} is well-known and easy but we believe including its proof may help the reader.
		\begin{lem}\label{fire}
			Let $k$ and $q$ be positive integers.
			Let $\Lambda$ be a $q$-ary $k$-stage CSR.
			Let $x=x_0x_1\dots$ be an infinite sequence
			produced by $\Lambda$.
			Then $x$ is eventually periodic,
			and the period of $\Lambda$ upon processing $x$ exists and is finite.
		\end{lem}
		\begin{proof}
			The infinite sequence $\Lambda^t(x_0,\dots,x_{k-1})$ for $0\le t<\infty$ takes values in the finite set $[q]^k$.
			Thus, by the pigeonhole principle, there exist $M$ and $N>0$ with
			\[
				\Lambda^M(x_0,\dots,x_{k-1}) = \Lambda^{M-N}(x_0,\dots,x_{k-1}).
			\]
			Let $n>M$. Then
			\begin{eqnarray*}
				(x_n,\dots,x_{n+k-1}) &=& \Lambda^{n}(x_0,\dots,x_{k-1})
			\\
				&=& \Lambda^{n-M}\Lambda^{M}(x_0,\dots,x_{k-1})\\
				&=& \Lambda^{n-M}\Lambda^{M-N}(x_0,\dots,x_{k-1})\\
				&=& \Lambda^{n-N}(x_0,\dots,x_{k-1})\\
				&=& (x_{n-N},\dots, x_{n-N+k-1}),
			\end{eqnarray*}
			hence
			\(
				x_n
				= x_{n-N}.
			\)
		\end{proof}

		We can now define LFSRs and $m$-sequences. As our computer results concern binary sequences, we take $q=2$.
		However, a higher level of generality would also be possible.
		\begin{df}
			Suppose a $k$-stage CSR $\Lambda$ produces the infinite output $x=x_0x_1\dots$
			and its feedback function is a linear transformation of $[q]$ when viewed as the finite field $\mathbb F_q$,
			where $q=2$.
			Then $\Lambda$ is a \emph{linear feedback shift register} (LFSR).
			Suppose the period $P$ of $\Lambda$ upon processing $x$ is $2^k-1$.
			Then $x_0\dots x_{P-1}$ is called an $m$-sequence (or maximal length sequence, or PN (pseudo-noise) sequence).
		\end{df}
		If $m$-sequences are pseudo-random in some sense then they should have high, or at least
		\emph{not unusually low}, complexity according to some measure.
		In 2015, Jason Castiglione (personal communication) suggested that \emph{automatic complexity} might be that measure.

		Our \emph{nondeterministic finite automata} will have no $\epsilon$-transitions, a unique start state and a set of accepting states.
		Without loss of generality for our purposes, the accepting state is unique.
		The language recognized by an automaton $M$ is the set $L(M)$ of words accepted by $M$.

		\begin{df}[{\cite{Kjos-EJC,MR1897300}}]\label{precise}
			Let $L(M)$ be the language recognized by the automaton $M$.
			Let $x$ be a sequence of finite length $n$.
			\begin{itemize}
				\item
					The (deterministic) \emph{automatic complexity} of $x$ is 
					the least number \(A(x)\) of states of a {deterministic finite automaton} \(M\) such that 
					\[
						L(M)\cap\{0,1\}^n = \{x\}.
					\]
				\item
					The \emph{nondeterministic automatic complexity} $A_N(x)$ is the minimum number of states of a
					nondeterministic finite automaton (NFA) $M$
					accepting $x$
					such that there is only one accepting path in $M$ of length $\abs{x}$.
				\item
					The \emph{non-total deterministic automatic complexity} $A^-(x)$
					is defined like $A(x)$ but without requiring totality of the transition function.
			\end{itemize}
		\end{df}
		As totality can always be achieved by adding at most one extra ``dead'' state, we have
		\[
			A_N(x)\le A^-(x)\le A(x)\le A^-(x) + 1.
		\]

		\begin{figure*}
			\centering
			\begin{tikzpicture}[shorten >=1pt,node distance=1.5cm,on grid,auto]
				\node[state,initial, accepting] (q_1)   {$q_1$}; 
				\node[state] (q_2)     [right=of q_1   ] {$q_2$}; 
				\node[state] (q_3)     [right=of q_2   ] {$q_3$}; 
				\node[state] (q_4)     [right=of q_3   ] {$q_4$};
				\node        (q_dots)  [right=of q_4   ] {$\ldots$};
				\node[state] (q_m)     [right=of q_dots] {$q_m$};
				\node[state] (q_{m+1}) [right=of q_m   ] {$q_{m+1}$}; 
				\path[->] 
					(q_1)     edge [bend left]  node           {$x_1$}     (q_2)
					(q_2)     edge [bend left]  node           {$x_2$}     (q_3)
					(q_3)     edge [bend left]  node           {$x_3$}     (q_4)
					(q_4)     edge [bend left]  node [pos=.45] {$x_4$}     (q_dots)
					(q_dots)  edge [bend left]  node [pos=.6]  {$x_{m-1}$} (q_m)
					(q_m)     edge [bend left]  node [pos=.56] {$x_m$}     (q_{m+1})
					(q_{m+1}) edge [loop above] node           {$x_{m+1}$} ()
					(q_{m+1}) edge [bend left]  node [pos=.45] {$x_{m+2}$} (q_m)
					(q_m)     edge [bend left]  node [pos=.4]  {$x_{m+3}$} (q_dots)
					(q_dots)  edge [bend left]  node [pos=.6]  {$x_{n-3}$} (q_4)
					(q_4)     edge [bend left]  node           {$x_{n-2}$} (q_3)
					(q_3)     edge [bend left]  node           {$x_{n-1}$} (q_2)
					(q_2)     edge [bend left]  node           {$x_n$}     (q_1)
				;
			\end{tikzpicture}
			\caption{
				A nondeterministic finite automaton that only accepts one sequence
				$x= x_1 x_2 x_3 x_4 \cdots x_n$ of length $n = 2m + 1$.
			}\label{fig1}
		\end{figure*}
		\begin{thm}[Hyde \cite{Kjos-EJC}]\label{Hyde}
			The nondeterministic automatic complexity $A_N(x)$ of a sequence $x$ of length $n$ satisfies
			\[
				A_N(x) \le {\lfloor} n/2 {\rfloor} + 1\text{.}
			\]
		\end{thm}
		Figure \ref{fig1} gives a hint to the proof of Theorem \ref{Hyde} in the case where $n$ is odd.
		Theorem \ref{Hyde} is sharp \cite{Kjos-EJC}, and
		experimentally we find that about 50\% of all binary sequences attain the bound.
		Thus, to ``fool'' finite automata this bound should be attained or almost attained.
	\section{Main result for FSRs}\label{theory}
		Our strategy will be to prove that if a sequence has low complexity, then it contains repeated parts, forcing any shift register producing it to
		be in the same state (including memory contents) at two distinct points in the sequence.

		We first introduce some automata theoretic notions that may not have standard names in the literature.
		\begin{df}
			\begin{itemize}
				\item[]
				\item A \emph{state sequence} is
					a sequence of states visited upon processing of an input sequence by a finite automaton.
				\item An \emph{abstract NFA} is an NFA without edge labels.
				\item The abstract NFA $M$ \emph{induced by} a state sequence $s=s_0\dots s_n$ is defined as follows.
					The states of $M$ are the states appearing in $s$.
					The transitions of $M$ are $s_i\to s_{i+1}$ for each $0\le i<n$.
				\item A state sequence $s=s_0\dots s_n$ is \emph{path-unique} if
					the abstract NFA induced by $s$ has only one path of length $\abs{s}$ from $s_0$ to $s_n$, namely $s$.
			\end{itemize}
		\end{df}
		We use the interval notation $s_{[i,j]}=s_i s_{i+1}\dots s_{j-1} s_j$ and we concatenate as follows: 
		${s_{[i,j]}}^{\frown}s_{[j,k]} = s_{[i,j]}s_{[j,k]} = s_{[i,k]}$.
		\begin{lem}\label{repetition-implies-power}
			Let $s=s_0\dots s_n$ be a path-unique state sequence.
			Suppose that $i\le j\le k$ are positive integers such that $s_i=s_j=s_k$, and
			$s_t \ne s_i$ for all $t\in [i,k]\setminus \{i,j,k\}$.
			Then
			$s_{[i,j]} = s_{[j,k]}$.
		\end{lem}
		\begin{proof}
			By uniqueness of path, $s_{[i,k]} = s_{[i,j]}s_{[j,k]} = s_{[j,k]}s_{[i,j]}$, so
			one of $s_{[i,j]}$ and $s_{[j,k]}$ is a prefix of the other.
			But considering the position of the second occurrence of $s_i$ in $s_{[i,k]}$, we can conclude
			$s_{[i,j]} = s_{[j,k]}$.
		\end{proof}
		\begin{df}
			Let $s=s_0\dots s_n$ be a path-unique state sequence and let $0\le i\le n$.
			The \emph{period} of $s_i$ in $s$ is defined to be
			$\min\{k-j: s_k=s_j=s_i, j<k\}$, if $s_i$ occurs at least twice in $s$, and to be $\infty$, otherwise.
		\end{df}
		An illustration of periods is given in Figure \ref{period}.

		\begin{lem}\label{few-states-preliminary}
			Let $s=s_0\dots s_n$ be a path-unique state sequence. If $i\le j$ and $t>0$ are integers such that
			$s_i=s_{i+t}$ and $s_j = s_{j+t}$, then
			$s_j\in \{s_i,\dots,s_{i+t}\}$.
		\end{lem}
		\begin{proof}
			Let $M$ be the abstract NFA induced by $s$.
 			We proceed by induction on the $k=k_j$ such that $j-t\in [i+(k-1)t, i+kt]$, which exists since $t>0$.
			If $k\le 0$ then $j\le i+t$ and we are done. So suppose
			$s_{j'}\in \{s_i,\dots,s_{i+t}\}$ for each $j'$ with $k_{j'}<k_j$.
			Both of the following state sequences of length $n+1$ are accepting for $M$:
			\begin{eqnarray*}
				\hat s &=& {s_{[0,i]}}^\frown\phantom{ {s_{[i,i+t]}}^\frown}
				           {s_{[i+t,j]}}^\frown {s_{[j,j+t]}}^\frown s_{[j,n]},\\
				s      &=& {s_{[0,i]}}^\frown {s_{[i,i+t]}}^\frown
				           {s_{[i+t,j]}}^\frown\phantom{ {s_{[j,j+t]}}^\frown} s_{[j,n]}.
			\end{eqnarray*}
			Since $s$ is path-unique, $s=\hat s$, and so $s_j = \hat s_j = s_{i+(j-(i+t))} = s_{j-t}$.
			Since $k_{j-t} = k_j - 1 < k_j$, by induction $s_{j-t}\in \{s_i,\dots,s_{i+t}\}$, giving
			$s_{j}\in \{s_i,\dots,s_{i+t}\}$, as desired.
		\end{proof}

		\begin{lem}\label{each-period-few-states}
			For each path-unique state sequence $s$, each number $t$ is the period of at most $t$ states in $s$.
		\end{lem}
		\begin{proof}
			We may of course assume $t<\infty$.
			Fix $i$ and suppose $t$ is the period of $s_i$.
			Let us count how many states $s_j$ there can be such that $t$ is the period of $s_j$.
			Since $t<\infty$, $s_i$ appears at least twice in $s$. Thus, either
			\begin{itemize}
				\item $i+t\le n$ and $s_i=s_{i+t}$, or
				\item $0\le i-t$ and $s_i=s_{i-t}$.
			\end{itemize}
			By Lemma \ref{few-states-preliminary}, either
			\begin{itemize}
				\item $s_j$ is among the states in $s_{[i,i+t]}$ and $s_i=s_{i+t}$, or
				\item $s_j$ is among the states in $s_{[i-t,i]}$ and $s_i=s_{i-t}$,
			\end{itemize}
			respectively. Either way, there are only at most $t$ choices of such $s_j$.
		\end{proof}

		\begin{lem}\label{analysis}
			Let $Q$ be a positive integer.
			Let $f:\{1,\dots,Q\}\to\mathbb N$ be a function such that $1\le f(1)$ and $f(i)<f(i+1)$ for each $i$.
			Then $i\le f(i)$ for each $i$.
		\end{lem}
		We omit the proof of the trivial Lemma \ref{analysis}.

		\begin{df}\label{IEEEref}
			Let $\alpha$ be a word of length $n$, and let $\alpha_i$ be the $i^{\text{th}}$ letter of $\alpha$ for $1\le i\le n$.
			We define the $u^{\text{th}}$ power of $\alpha$ for certain values of $u\in\mathbb Q_{\ge 0}$ (the set of nonnegative rational numbers) as follows.
			\begin{itemize}
				\item If $u=0$ then $\alpha^u$ is the empty word.
				\item If $u>0$ is an integer then the power $\alpha^u$ is defined inductively by $\alpha\, \alpha^{u-1}$, where juxtaposition denotes concatenation.
				\item If $u=v+k/n$ where $0<k<n$, and $k$ is an integer,
					then $\alpha^u$ denotes $\alpha^v \alpha_1\dots\alpha_k$.
			\end{itemize}
		\end{df}
		As an example of Definition \ref{IEEEref}, we have $ABBA^{1.5}=ABBAAB$.
		\begin{lem}\label{wordToOccurrence}
			Let $f:\mathbb N\to\mathbb N$, $n\ge 0$, and $u\in\mathbb Q_{\ge 0}$.
			Suppose that all $u^\text{th}$ powers $\alpha^u$ within a sequence $x$ of length $n$ satisfy $u\le f(\abs{\alpha})$,
			where $f$ is non-increasing.
			Let $s$ be a path-unique state sequence.
			Let $q_1,\dots,q_Q$ be a list of states of $s$ ordered by increasing period.
			Let $a_i$ be the number of occurrences of $q_i$ in $s$.
			Let $M_s$ be the abstract NFA induced by $s$.

			Suppose moreover that $x$ and $s$ are related as follows:
			$x$ is the input read along the unique accepting path of length $\abs{x}$
			of some NFA $M$ which is obtained from $M_s$ by assigning one label to each edge.

			Let
			\begin{eqnarray*}
				\left(b_1,b_2,\dots\right) = (f(1) + 1, f(2) + 1, f(2) + 1,
				\dots,\\
				\underbrace{f(i)+1,\dots,f(i)+1}_{\text{$i$ times}},\dots).
			\end{eqnarray*}
			Then $a_i\le b_i$ for each $i$.
		\end{lem}
		\begin{proof}
			For each $1\le i\le Q$, let $\ell_i$ be the period of $q_i$.
			(For instance, we could have $(\ell_1,\ell_2,\dots)=(3,4,4,4,4,5,6)$.)
			If $q_i$ occurs $u+1$ times then by Lemma \ref{repetition-implies-power}
			it occurs during the processing of a $u$th power $\alpha^u$ where $\abs{\alpha}=\ell_i$.
			Thus $q$ occurs at most $f(\ell_i)+1$ times, i.e., $a_i\le f(\ell_i)+1$.
			By Lemma \ref{each-period-few-states},
			the sequence $(\ell_1,\ell_2,\ell_3,\dots)$ is a subsequence of the sequence
			$(1,2,2,3,3,3,\dots)$ hence by Lemma \ref{analysis}, dominates it pointwise.
			And so $a_i\le f(\ell_i)+1\le b_i$.
		\end{proof}

		In particular, Lemma \ref{wordToOccurrence} tells us that if $x$ is square-free then
		each state can occur at most twice,
		which was observed by Shallit and Wang \cite{MR1897300}.
		\begin{lem}\label{domination}
			Let $s=s_0\dots s_n$ be a state sequence.
			Let $q_1,\dots,q_Q$ be the distinct states appearing in $s$, in any order.
			Let $a_i\ge 1$ be the number of times $q_i$ occurs.
			Let $T=n+1=\abs{s}= \sum_{i=1}^Q a_i$.
			Let $Q_0\le Q$ and let $g:\mathbb Z_{\ge 0}\to\mathbb Z_{\ge 0}$.
			If
			$a_i\le g(i)$ for all $1\le i\le Q$, and $g(i)=2$ for all $Q_0<i<\infty$, with
			\(
				T_0:=\sum_{i=1}^{Q_0}g(i) \le T,
			\)
			then
			\[
				Q\ge Q_0 + \ceil{\frac{T-T_0}2}.
			\]
		\end{lem}
		\begin{proof}
			Let $w$ be such that $T-T_0\in\{2w+1,2w+2\}$, i.e.,
			$w=\ceil{(T-T_0)/2}-1$.
			Then we want to show $Q\ge Q_0+w+1$.
			If $Q<Q_0+w+1$ then $Q\le Q_0+w$ and then
			\[
				T = \sum_{i=1}^Q a_i
				\le \sum_{i=1}^{Q_0}g(i) + \sum_{i=Q_0+1}^{Q_0+w} 2
				= T_0 + 2w,
			\]
			so $2w+1\le T-T_0\le 2w$, a contradiction.
		\end{proof}
		\begin{lem}\label{shortPowers}
			Let $k$ be a positive integer.
			Let $\Lambda$ be a $k$-stage CSR.
			Let $x=x_0x_1\dots$ be an infinite sequence
			produced by $\Lambda$.
			Let $P$ be the period of $\Lambda$ upon processing $x$.
			Suppose a sequence $\alpha$ of length $\ell<P$ is repeated $u$ times consecutively within $x$,
			i.e., $\alpha^u$ is a contiguous subsequence of $x$.

			Then $u<k/\ell+1$, i.e., $u\le \ceil{\frac{k}\ell}$, i.e., $u\le f(\abs{\alpha})$
			where
			$f(a)=\ceil{k/a}$.
		\end{lem}
		\begin{proof}
			Suppose to the contrary that $x$ contains a block
			\[
				x_j\dots x_{j+\ell u-1} = z_1\dots z_{\ell u} = y_1\dots y_\ell y_1\dots y_\ell \dots
			\]
			with $u$ many blocks of length $\ell$, where $(u-1)\ell\ge k$, i.e., $\ell+k\le \ell u$.
			Let $q\ge 0$ and $r\ge 0$ be such that $k=q\ell+r$.
			We have 
			\[
				\Lambda^j(x_0,\dots, x_{k-1})
				= (x_j\dots x_{j+k-1})
				= (z_1,\dots,z_k)
			\]
			\[
			= \overbrace{(y_1\dots y_\ell) \dots (y_{1}\dots y_{\ell})}^{q\text{ times}} y_1\dots y_r
				= (z_{\ell+1},\dots,z_{\ell+k})
			\]
			\[
				= (x_{j+\ell}\dots x_{j+\ell+k-1})
				= \Lambda^{j+\ell}(x_0,\dots,x_{k-1}).
			\]
			So $x$ is eventually periodic with period $N\le a < P$, a contradiction.
		\end{proof}
		\begin{thm}\label{mainy}\label{LFSRs-fool}
				Let $x$ be an $m$-sequence and let $n=\abs{x}$.
				Then $n/2-A_N(x)=O(\log^2(n))$.
		\end{thm}
		\begin{proof}
			Note that if $x$ is produced by a $k$-stage CSR $\Lambda$, then
			the period $P$ of $\Lambda$ upon processing $x$ is just $P=n$.

			Let $Q=A_N(x)$. Thus $Q$ is the number of states of
			an NFA $M$ with only one accepting path $s$ of length $n$, accepting $x$ along that path.
			Let $q_1,\dots,q_Q$ be the states of $M$ ordered by increasing period within $s$.
			Let $a_i$ be the number of occurrences of $q_i$.

			By Lemma \ref{shortPowers},
			if $x$ contains $\alpha^u$ where
			$1\le \abs{\alpha}\le k<P$,
			then $u\le f(\abs{\alpha})$ where $f(a)=\ceil{k/a}$, a non-increasing function.
			By Lemma \ref{wordToOccurrence}, each $a_i\le b_i$, where
			\begin{eqnarray*}
				\left(b_1,b_2,\dots\right)
				= (f(1) + 1, f(2) + 1, f(2) + 1,
				\dots,\\
				\underbrace{f(i)+1,\dots,f(i)+1}_{\text{$i$ times}},\dots).
			\end{eqnarray*}
			Let $T=n+1=\abs{s}= \sum_{i=1}^Q a_i$.
			Let $g(i)=\max\{b_i,2\}$.
			Let $Q_0$ be the least integer such that
			$b_i\le 2$ for all $i>Q_0$.
			Then since $f(k)+1=\ceil{\frac{k}{k}}+1= 2$ and since $2k(k-1)\le n+1$,
			\[
				T_0:=\sum_{i=1}^{Q_0}g(i) = \sum_{i=1}^{Q_0}b_i\le \sum_{i=1}^{k-1} i \left(\ceil{\frac{k}{i}} + 1\right)
			\]
			\[
				\le \sum_{i=1}^{k-1} i\left(\frac{k}{i}+2\right) = k(k-1) + k(k-1)
			\]
			\[
				= 2k(k-1) \le n+1 = T,
			\]
			and \(g(i)=2\) for all $i>Q_0$.
			Hence by Lemma \ref{domination},
			\(
				Q\ge Q_0 + Q_1,
			\)
			where $Q_1=\ceil{\frac{T-T_0}2}$. Note that $Q_1$ is the minimum number of twos whose sum is at least $T-T_0$.
			(For instance, if $T-T_0=2w+1$, say, then $Q_1=w+1=\ceil{\frac{T-T_0}2}$.)

			Thus
			\[
				\overbrace{
				\left(\ceil{\frac{k}1} + 1\right)
				+ 2\left(\ceil{\frac{k}2} +1\right)
				+\dots
				+ (k-1)\left(\ceil{\frac{k}{k-1}} +1\right)
				}^{\text{$Q_0$ many terms}}
			\]
			\[
				+ \overbrace{2+2+\dots}^{Q_1\text{ many terms}} \ge n+1.
			\]
			Clearly, $Q_0=\sum_{i=1}^{k-1} i=k(k-1)/2$.
			Now $T_0+(T-T_0)=n+1$, $T_0\le 2k(k-1)$, and $2Q_1\ge T-T_0$, so 
			\[
				2k(k-1) + 2Q_1 \ge n+1,\qquad
				Q_1 \ge \frac{n+1}2 - k(k-1),
			\]
			and
			\[
				Q  \ge Q_0 + Q_1 \ge \frac{k(k-1)}2 + \frac{n+1}2 - k(k-1)
			\]
			\[
				= \frac{n+1}2 - \frac{k(k-1)}2.
			\]
			Thus
			\[
				A_N(x)\ge \frac{n+1}2 - \frac{\log_2(n+1)(\log_2(n+1)-1)}2.
			\]
		\end{proof}

	\section{Computer results}\label{computer}
		\subsection{Linear FSRs}
			\begin{thm}\label{christmas}
				Let $x\in\{0,1\}^n$ be an $m$-sequence, where $n=2^k-1$, $k\le 5$.
				Then $A_N(x)=\floor{n/2}+1$.
			\end{thm}
			Theorem \ref{christmas} was verified in 36 hours using a Python script.
			\begin{thm}\label{refute-NSF}
				There exists a sequence $x$ with $A^-(x)-A_N(x)\ge 2$.
				In fact, there is an $m$-sequence $x$ with $A^-(x)-A_N(x)=2$.
			\end{thm}
			\begin{proof}
				Let
				\(
					x = 0001010110100001100100111110111.
				\)
				A computer run showed that $A^-(x)\ge 18$.
				The production of this sequence by an LFSR with 5 bits is shown in detail in Figure \ref{operation}.
				Figure \ref{period} can be used to verify that $A^-(x)\le 18$.
				According to Theorem \ref{christmas}, $A_N(x)=16$.
			\end{proof}
		\begin{figure*}
			\begin{subfigure}[b]{0.5\textwidth}
				\centering
				\tikzset{every state/.style={minimum size=0pt}}
				\begin{tikzpicture}[shorten >=1pt,node distance=0.9cm,on grid,auto]
	        		\node[state, initial] (q_0) {$0$};
					\node[state] (q_1) [right=of q_0]   {$1$}; 
					\node[state, accepting] (q_2)     [right=of q_1   ] {$2$}; 
					\node[state] (q_3) [below=of q_2] {$3$}; 
					\node[state] (q_4) [below=of q_3] {$4$};
					\node[state] (q_5) [below=of q_4] {$5$};
					\node[state] (q_6) [below=of q_5] {$6$};
					\node[state] (q_7) [below=of q_6] {$7$}; 
					\node[state] (q_8) [below=of q_7] {$8$}; 
					\node[state] (q_9) [below=of q_8] {$9$}; 
					\node[state] (q_10) [below=of q_9] {$\mathrm{A}$}; 
					\node[state] (q_11) [below=of q_10] {$\mathrm{B}$}; 
					\node[state] (q_12) [below=of q_11] {$\mathrm{C}$}; 
					\node[state] (q_13) [below=of q_12] {$\mathrm{D}$}; 
					\node[state] (q_14) [below=of q_13] {$\mathrm{E}$}; 
					\node[state] (q_15) [below=of q_14] {$\mathrm{F}$}; 
					\node[state] (q_16) [below=of q_15] {$\mathrm{G}$}; 
					\node[state] (q_17) [below=of q_16] {$\mathrm{H}$}; 
					\path[->] 
						(q_0)  edge [bend left]           node {\color{red}0} (q_1)
						(q_1)  edge [bend left]           node {\color{red}0} (q_2)
						(q_2)  edge [bend left]           node {\color{red}0} (q_3)
						(q_3)  edge [bend left]           node {\color{red}1} (q_4)
						(q_4)  edge [bend left]           node {\color{red}0} (q_5)
						(q_5)  edge [bend left]           node {\color{red}1} (q_6)
						(q_6)  edge [bend left]           node {\color{red}0} (q_7)
						(q_7)  edge [bend left]           node {\color{red}1} (q_8)
						(q_8)  edge [bend left]           node {\color{red}1} (q_9)
						(q_9)  edge [bend left]           node {\color{red}0} (q_10)
						(q_10) edge [bend left]           node {\color{red}1} (q_11)
						(q_11) edge [bend left]           node {\color{red}0} (q_12)
						(q_12) edge [bend left]           node {\color{red}0} (q_13)
						(q_13) edge [bend left]           node {\color{red}0} (q_14)
						(q_14) edge [bend left]           node {\color{red}0} (q_15)
						(q_15) edge [bend left]           node {\color{red}1} (q_16)
						(q_16) edge [bend left]           node {\color{red}1} (q_17)
						(q_17) edge [loop below]          node {\color{red}0} (q_17)
						(q_17) edge [bend  left=70]       node {\color{red}1} (q_15)
						(q_15) edge [bend right=70,right] node {\color{red}0} (q_13)
						(q_14) edge [bend  left=70]       node {\color{red}1} (q_12)
						(q_12) edge [bend right=70,right] node {\color{red}1} (q_10)
						(q_11) edge [bend  left=70]       node {\color{red}1} (q_9)
						(q_9)  edge [bend right=70,right] node {\color{red}1} (q_7)
						(q_7)  edge [bend  left=70]       node {\color{red}0} (q_5)
						(q_6)  edge [bend right=70,right] node {\color{red}1} (q_4)
						(q_4)  edge [bend  left=70]       node {\color{red}1} (q_2)
					;
				\end{tikzpicture}
				\caption{}
			\end{subfigure}
			\begin{subfigure}[b]{0.5\textwidth}
				\centering
				\begin{tabular}{|r|c|c|}
					\hline
					Time & State & Period of state\\
					\hline
					 0 & 0 & $\infty$\\
					 1 & 1 & $\infty$\\
					 2 & 2 & 29\\
					 3 & 3 & $\infty$\\
					 4 & 4 & 26\\
					 5 & 5 & 23\\
					 6 & 6 & 23\\
					 7 & 7 & 20\\
					 8 & 8 & $\infty$\\
					 9 & 9 & 17\\
					10 & A & 14\\
					11 & B & 14\\
					12 & C & 11\\
					13 & D & 8\\
					14 & E & 8\\
					15 & F & 5\\
					16 & G & $\infty$\\
					17 & H & 1\\
					\hline
					18 & H\\
					19 & H\\
					20 & F\\
					21 & D\\
					22 & E\\
					23 & C\\
					24 & A\\
					25 & B\\
					26 & 9\\
					27 & 7\\
					28 & 5\\
					29 & 6\\
					30 & 4\\
					31 & 2\\
				\end{tabular}
				\caption{}
			\end{subfigure}
			\caption{
				An optimal deterministic automaton, witness to $A^-({\color{red}0001010110100001100100111110111})=18$, and its
				times, states, and periods.
				There is only 1 state with period 1, and in general at most $\ell$ states with period $\ell$.
			}\label{period}
		\end{figure*}
		\begin{figure}
			\begin{subfigure}[b]{0.5\textwidth}
				\centering
				\(
					\begin{bmatrix}
						1&1&1&0&1\\
						1&0&0&0&0\\
						0&1&0&0&0\\
						0&0&1&0&0\\
						0&0&0&1&0\\
					\end{bmatrix}
					\begin{bmatrix}
						{\color{blue}0}\\
						{\color{blue}1}\\
						{\color{blue}0}\\
						{\color{blue}0}\\
						{\color{red}0}\\
					\end{bmatrix}
					=
					\begin{bmatrix}
						{\color{green}1}\\
						{\color{green}0}\\
						{\color{green}1}\\
						{\color{green}0}\\
						{\color{red}0}\\
					\end{bmatrix}
				\)
				\caption{One LFSR step as a matrix multiplication.}
			\end{subfigure}
			\begin{subfigure}[b]{0.5\textwidth}
				\centering
				\begin{tabular}{ccccc}
					$\oplus$&$\oplus$&$\oplus$&\phantom{$\oplus$} &$\oplus$ \\
					\hline
					{\color{blue}0} & {\color{blue}1} & {\color{blue}0} & {\color{blue}0} & {\color{red}0}\\
					{\color{green}1}&{\color{green}0}&{\color{green}1}&{\color{green}0}&{\color{red}0}\\
					0&1&0&1&{\color{red}0}\\
					1&0&1&0&{\color{red}1}\\
					1&1&0&1&{\color{red}0}\\
					0&1&1&0&{\color{red}1}\\
					1&0&1&1&{\color{red}0}\\
					0&1&0&1&{\color{red}1}\\
					0&0&1&0&{\color{red}1}\\
					0&0&0&1&{\color{red}0}\\
					0&0&0&0&{\color{red}1}\\
					1&0&0&0&{\color{red}0}\\
					1&1&0&0&{\color{red}0}\\
					0&1&1&0&{\color{red}0}\\
					0&0&1&1&{\color{red}0}\\
					1&0&0&1&{\color{red}1}\\
					0&1&0&0&{\color{red}1}\\
					0&0&1&0&{\color{red}0}\\
					1&0&0&1&{\color{red}0}\\
					1&1&0&0&{\color{red}1}\\
					1&1&1&0&{\color{red}0}\\
					1&1&1&1&{\color{red}0}\\
					1&1&1&1&{\color{red}1}\\
					0&1&1&1&{\color{red}1}\\
					1&0&1&1&{\color{red}1}\\
					1&1&0&1&{\color{red}1}\\
					1&1&1&0&{\color{red}1}\\
					0&1&1&1&{\color{red}0}\\
					0&0&1&1&{\color{red}1}\\
					0&0&0&1&{\color{red}1}\\
					1&0&0&0&{\color{red}1}\\
		            \hline
	            	{\color{blue}0}&{\color{blue}1}&{\color{blue}0}&{\color{blue}0}&0\\
				\end{tabular}
				\caption{Producing an $m$-sequence.}
			\end{subfigure}
			\caption{The operation of a
				linear feedback shift register producing the sequence from Theorem \ref{refute-NSF}.
			}\label{operation}
		\end{figure}
			We also found another $m$-sequence $y$ for $k=5$ with $A^-(y)=17$.
			Thus not every $m$-sequence has maximal $A^-$-complexity:
			\begin{thm}
				There is an $m$-sequence $x$ and a sequence $y$ with $\abs{x}=\abs{y}$ such that $A^-(x)<A^-(y)$.
			\end{thm}
			\begin{con}
				There is an $m$-sequence $x$ and a sequence $y$ with $\abs{x}=\abs{y}$ such that $A_N(x)<A_N(y)$.
			\end{con}
			Using our current algorithm and implementation,
			the calculation of $A_N(x)$ for $m$-sequences $x$ of length $2^6-1$ is unfortunately out of reach.
		\subsection{Nonlinear FSRs}
			For $k=3$ there are two possible feedback functions that give an injective function with a single cycle,
			\[
				F(p,q,r)=p+pq+r+1\text{ and }F(p,q,r)=q+pq+r+1\quad\text{mod 2}.
			\]
			One of them gives the output $00011101$, which has $A_N(00011101)=4$ and so is \emph{not} maximally $A_N$-complex.

		\section{Relation to subword complexity}\label{lfsr}
		\begin{df}
			A word $w$ is a \emph{factor}, or \emph{contiguous subsequence}, of a word $v$ if $v=awb$ for some words $a$, $b$.
			For a finite word $x$, $x^\infty$ is the infinite word satisfying $x\,x^\infty = x^\infty$.
			For a finite or infinite word $u$, $p_{u}(k)$ is the number of distinct factors of $u$ of length $k$.
			The cyclic subword complexity of $x$ is
			\[
				(p_{x^\infty}(1),\dots,p_{x^\infty}(\abs{x})).
			\]
			The plain subword complexity of $x$ is
			\[
				(p_{x}(1),\dots,p_{x}(\abs{x})).
			\]
		\end{df}
		In general, neither of maximum subword complexity and maximum $A_N$-complexity implies the other.
		\begin{thm}
			Maximum subword complexity can be characterized as follows.
			\begin{enumerate}
				\item The cyclic subword complexity of a $b$-ary word is pointwise bounded above by
					\[
						(b^1, b^2, \dots, b^t,n,n,\dots,n)
					\]
					where $t$ is maximal such that $b^t\le n$.
				\item This upper bound is realized by $m$-sequences when $n=b^k-1$, $k\ge 0$, $b=2$.
			\end{enumerate}
		\end{thm}
		\begin{proof}
			(1) is because both
			\[
				(b^1,b^2,\dots,b^n)
			\]
			and
			\[
				(n,n,\dots,n)
			\]
			are upper bounds, and the pointwise minimum of two upper bounds is an upper bound.
			To prove (2), we need to show that an $m$-sequence $x$ has
			\[
				(p_{x^\infty}(1),\dots,p_{x^\infty}(n))
			\]
			\[
				= (b^1,b^2,\dots,b^{k-1},b^k-1, b^k-1,\dots, b^k-1),
			\]
			that is,
			\[
				p_{x^\infty}(i) =
				\begin{cases}
					b^i     & \text{if }i\le k-1,\\
					b^k - 1 & \text{if }k\le i\le b^k-1.
				\end{cases}
			\]
			It suffices to show $p_{x^\infty}(k)=b^k-1$, since
			\begin{itemize}
				\item[(i)] this gives $p(i)=b^i$ for $i\le k-1$ (only one string of length $k$ is missing, so all strings of shorter length must be present; since any missing string of shorter length would give at least $b>1$ missing strings of length $k$), and
				\item[(ii)] $p(i)$ is monotonically increasing with $i$ (if two words have distinct prefixes of a certain length, then the strings are distinct).
			\end{itemize}
			The statement $p(k)=b^k-1$ when $b=2$ follows easily from a note labelled (4) in \cite{newwave},
			namely
			\begin{quote}
				``4. A sliding window of length $k$, passed along an $m$-sequence for $2^k - 1$ positions,
				will span every possible $k$-bit number, except all zeros, once and only once.
				That is, every state of a $k$-bit state register will be encountered, with the exception of all zeros.''
			\end{quote}
			This statement (4) surely is already implicit in Golomb's monograph. In any case, it is almost immediate from the fact that $m$ bits are saved in the state and the sequence is maximum-length.
			This property of $m$-sequences, and the fact that one can work with an arbitrary finite field in place of \(\{0,1\}\), is also explicitly mentioned in \cite[Theorem 7.43]{MR860948}.
		\end{proof}

		\section{The no-long-high-powers property}

			\begin{df}
				Let $k\ge 0$ and $u\in\mathbb Q_{\ge 0}$.
				The no-long-high-power (NLHP or ``no LHP'') property of a sequence $x$ of length $2^k-1$ says that if a word $\alpha$ of length $s<2^k-1$
				is such that $\alpha^u$ is a factor of $x^\infty$, then $u<k/s+1$.
			\end{df}
			By Lemma \ref{shortPowers}, $m$-sequences have the NLHP property.

		\begin{thm}
			Let $x$ be a word of length $n=2^k-1$. The following are equivalent:
			\begin{enumerate}
				\item[(i)] $x$ has the NLHP property.
				\item[(ii)] $x$ has maximal cyclic subword complexity.
			\end{enumerate}
		\end{thm}
		\begin{proof}
			Recall that when $n=2^k-1$ then the maximal subword complexity is realized by $m$-sequences and is
			\begin{equation}
				(2,4,8,\dots,2^{k-1},n,n,\dots,n).\label{enn}
			\end{equation}
			Let us prove that (ii) implies (i).
			Suppose that $w^{u}$ is contained cyclically (using only a single trip through the cycle) in $x$.
			We need to show that $\abs{w^{u-1}}<k$ where $n=2^k-1$.
			It's just that if there is an LHP then there are two positions giving the same subword, thereby reducing one of the $n$'s in (\ref{enn}) to $n-1$.

			Let us now prove that (i) implies (ii). We need to show that if, say, 01 is not a factor of $x^\infty$ then there are so many occurrences of 00, 10, 11 as factors as to make an LHP.
			If there is no 01 then there are many strings of length $k$ that are missing, and so some are repeated.
			Thus, if $p_{x^\infty}(s)<2^s$ for $s<t$ then also $p_{x^\infty}(t)<2^t$.
			Thus, $p_{x^\infty}(k-1)<2^{k-1}$.
			And then we can argue that $p_{x^\infty}(k)<2^k-1$, as well. So by the Pigeonhole Principle some word of length $k$ is
			repeated and hence there is an LHP.
		\end{proof}

		\begin{cor}\label{lics}
			Theorem \ref{mainy}  applies to any word of maximal subword complexity (when the length is $2^k-1$).
		\end{cor}
		Corollary \ref{lics} is of interest because of the following result.
		\begin{thm}
			Words of maximal subword complexity do not in general have maximum $A_N$-complexity.
		\end{thm}
		\begin{proof}
			It is easily checked that already at length 6, we have a string of maximal subword complexity but not maximal $A_N$-complexity:
			namely 001100.
		\end{proof}
		Thus, while experimentally our computer results suggest that $m$-sequences always have maximal $A_N$-complexity,
		our main theoretical result Theorem \ref{mainy} show that $m$-sequences have fairly high $A_N$-complexity also applies to some sequences that
		demonstrably do not have maximal $A_N$-complexity.

	\section*{Acknowledgments}
		This work
		was partially supported by
		a grant from the Simons Foundation (\#315188 to Bj\o rn Kjos-Hanssen).
		This material is based upon work supported by the National Science Foundation under Grant No.\ 1545707.
	\bibliographystyle{elsarticle-num}
	\bibliography{LFSR-DM}
\end{document}